\documentclass[12pt]{elsarticle}

\usepackage{amsfonts,amsmath,amssymb,mathrsfs,amsthm}
\usepackage{epsfig,epstopdf,graphicx,graphics}
\usepackage{color}
\usepackage{braket}
\usepackage{float}
\usepackage{ifpdf}
\usepackage[colorlinks=true]{hyperref}

\newtheorem{theorem}{Theorem}[section]

\newtheorem{corollary}{Corollary}[section]

\newtheorem{definition}{Definition}[section]
\newtheorem{example}{Example}[section]

\newtheorem{lemma}{Lemma}[section]

\newtheorem{problem}{Problem}[section]
\newtheorem{proposition}{Proposition}[section]
\newtheorem{remark}{Remark}[section]

\newcommand{\bthm}{\begin{theorem}}
\newcommand{\ethm}{\end{theorem}}
\newcommand{\blem}{\begin{lemma}}
\newcommand{\elem}{\end{lemma}}
\newcommand{\bex}{\begin{example}}
\newcommand{\eex}{\end{example}}
\newcommand{\bprop}{\begin{proposition}}
\newcommand{\eprop}{\end{proposition}}
\newcommand{\bplm}{\begin{problem}}
\newcommand{\eplm}{\end{problem}}
\newcommand{\bmrk}{\begin{remark}}
\newcommand{\emrk}{\end{remark}}
\newcommand{\bdfn}{\begin{definition}}
\newcommand{\edfn}{\end{definition}}
\newcommand{\bcor}{\begin{corollary}}
\newcommand{\ecor}{\end{corollary}}

\newcommand{\beq}{\begin{equation}}
\newcommand{\eeq}{\end{equation}}
\newcommand{\beqm}{\begin{equation*}}
\newcommand{\eeqm}{\end{equation*}}
\newcommand{\beqn}{\begin{eqnarray}}
\newcommand{\eeqn}{\end{eqnarray}}
\newcommand{\beqnm}{\begin{eqnarray*}}
\newcommand{\eeqnm}{\end{eqnarray*}}
\newcommand{\bea}{\begin{align}}
\newcommand{\eea}{\end{align}}
\newcommand{\bead}{\begin{aligned}}
\newcommand{\eead}{\end{aligned}}
\newcommand{\beam}{\begin{align*}}
\newcommand{\eeam}{\end{align*}}
\newcommand{\bs}{\begin{subequations}}
\newcommand{\es}{\end{subequations}}

\newcommand{\bei}{\begin{itemize}}
\newcommand{\eei}{\end{itemize}}
\newcommand{\bed}{\begin{description}}
\newcommand{\eed}{\end{description}}
\newcommand{\bee}{\begin{enumerate}}
\newcommand{\eee}{\end{enumerate}}

\newcommand{\bey}{\begin{array}}
\newcommand{\eey}{\end{array}}

\newcommand{\beb}{\begin{thebibliography}{99}}
\newcommand{\eeb}{\end{thebibliography}}

\newcommand{\mbf}{\mathbf}

\newcommand{\la}{\label}

\def\mm[#1]{{\rm #1}}
\def\f[#1]{\frac1{\sqrt{2^{#1}}}}

\usepackage{scalerel,stackengine}
\stackMath
\newcommand\reallywidehat[1]{%
\savestack{\tmpbox}{\stretchto{%
  \scaleto{%
    \scalerel*[\widthof{\ensuremath{#1}}]{\kern-.6pt\bigwedge\kern-.6pt}%
    {\rule[-\textheight/2]{1ex}{\textheight}}
  }{\textheight}%
}{0.5ex}}%
\stackon[1pt]{#1}{\tmpbox}%
}

\begin{document}

\author[mymainaddress]{Zhiyuan Dong}
\ead{dongzhiyuan@hit.edu.cn}
\address[mymainaddress]{School of Science, Harbin Institute of Technology, Shenzhen, China}

\author[mysecondaryaddress]{Weichao Liang}
\ead{weichao.liang@xjtu.edu.cn}
\address[mysecondaryaddress]{School of Automation Science and Engineering, Faculty of Electronic and Information Engineering, Xi'an Jiaotong University, China}

\author[mythirdaddress,myfourthaddress]{Guofeng Zhang}
\ead{guofeng.zhang@polyu.edu.hk}
\address[mythirdaddress]{Department of Applied Mathematics, The Hong Kong Polytechnic University, Hung Hom, Kowloon, Hong Kong SAR, China}
\address[myfourthaddress]{Shenzhen Research Institute, The Hong Kong Polytechnic University, Shenzhen, China}

\begin{abstract}
We establish a framework for realizing back-action-evading (BAE) measurements and quantum non-demolition (QND) variables in linear quantum systems. The key condition, a purely imaginary Hamiltonian with a real or imaginary coupling operator, enables BAE measurements of conjugate observables. Symmetric coupling further yields QND variables. For non-compliant systems, coherent feedback can engineer BAE measurements. Crucially, the QND interaction condition simultaneously ensures BAE measurements and promotes the coupling operator to a QND observable. This work provides a unified structural theory for enhancing precision in quantum metrology and sensing.    
\end{abstract}


\title{Back-Action-Evading Measurements and Quantum Non-Demolition Variables via Linear Systems Engineering}

\maketitle


\tableofcontents

\section{Introduction}

In a typical indirect measurement within a quantum system \cite{belavkin1989nondemolition,Belavkin1994,HMW95,BvHJ07}, a auxiliary quantum system, often termed a probe, is deliberately coupled to the system of interest. Information about the latter is inferred by performing a direct measurement on the probe after their interaction. In this framework, the probe serves as an input channel before the interaction and transforms into an output channel afterwards. A fundamental consequence of this information acquisition process is the introduction of measurement back-action, which inevitably disturbs the dynamics of the system of interest.

A key objective in quantum control and metrology is to identify and protect parts of a quantum system that are inherently immune to such measurement-induced disturbances. If a specific subsystem of a linear quantum system is inaccessible to the probe—meaning it is neither influenced by the input probe nor detectable in the output—it is termed a decoherence-free subsystem (DFS) \cite{NY13,PDP17,ZGPG18,ZPL20}. Importantly, a DFS is not an isolated system; it interacts with other parts of the composite system. Its decoherence-free property implies that its state evolves unitarily and remains separable from the rest of the system's degrees of freedom, provided the system-environment coupling respects a certain symmetry. In the Heisenberg picture, this concept is closely related to the decoherence-free subspaces in the Schrödinger picture \cite[Sec. III-C]{DZWW23}, which play a crucial role in protecting quantum information in quantum computation \cite{TV08}.

While a DFS represents a structural immunity where certain system degrees of freedom are statically inaccessible to the probe, a more dynamical form of back-action evasion can be engineered at the level of input-output relationships. Specifically, a quantum back-action evading (BAE) measurement is achieved when a particular observable of the output probe, denoted $\boldsymbol{y}_q$, is insensitive to a specific conjugate observable of the input probe, denoted $\boldsymbol{u}_p$. In the language of linear systems and the Kalman canonical form, this condition corresponds to a zero transfer function from the input $\boldsymbol{u}_p$ to the output $\boldsymbol{y}_q$. This frequency-domain characterization is consistent with the standard linear systems approach to quantum measurement theory, where transfer functions have been rigorously established as a valid tool for describing BAE conditions \cite[Fig. 2]{NY14}. In continuous-wave BAE experiments, the quantity of interest is typically the measured quadrature of the output signal \cite{Liu_2022}; a zero transfer function in the frequency domain guarantees that the back-action noise from input $\boldsymbol{u}_p$ is evaded across all relevant frequencies. Quantum BAE measurements are therefore of paramount importance in quantum sensing and metrology \cite{WC13,BQND24}, as they enable precision beyond the standard quantum limit by selectively rendering the measurement immune to specific fluctuations of the input field.

The concept of back-action evasion focuses on the input-output relation—evading back-action from a specific input channel. A related but distinct paradigm concerns the system observable itself: the notion of quantum non-demolition (QND) measurements \cite{BK96}. According to the Heisenberg uncertainty principle, the precision of simultaneous measurements of canonically conjugate observables (e.g., position and momentum, or amplitude and phase) is fundamentally limited. Consequently, an indirect measurement of an observable $\mathcal{O}$ generally introduces back-action onto its conjugate counterpart. However, if the Heisenberg-picture evolution of $\mathcal{O}$ is a function only of itself and is independent of its conjugate observables, then the measurement of $\mathcal{O}$ does not perturb its own future evolution. This defines a QND measurement: an observable $\mathcal{O}$ is a QND variable if it can be measured repeatedly, with the measurement back-action confined entirely to its conjugate variables, leaving the measured observable's future trajectory unaffected \cite{TC10,WC13,NY14,ZGPG18,LOW+21}. While back-action evasion ensures that a specific output is free from a specific input noise, QND measurement ensures that the observable of interest is free from the measurement back-action generated by its own measurement.

This paper presents a systematic framework for realizing BAE measurements and identifying QND variables using linear systems engineering. By leveraging the structural properties of linear quantum systems, we unify the concepts of BAE and QND within a common state-space representation, providing both a theoretical foundation and practical design criteria for back-action evasion in continuous-mode measurement scenarios.
Some commonly notation used in this paper are listed as follows.

\textit{Notation}.
\begin{itemize}
\item $\imath =\sqrt{-1}$ is the imaginary unit. $I_{k}$ is the identity matrix and $0_{k}$ the zero matrix in $\mathbf{C}^{k \times k}$. $\delta_{ij}$ denotes the Kronecker delta;
i.e.,~$I_k=[\delta_{ij}]$. $\delta(t)$ is the Dirac delta function. ${\rm Re}(X)$ denotes the real part of the matrix $X$; while ${\rm Im}(X)$ is its imaginary part.

\item $x^{\ast}$ denotes the complex conjugate of a complex number $x$ or
the adjoint of an operator $x$. Clearly. $(xy)^\ast = y^\ast x^\ast$.  Given two operators $\bf{x}$ and $\bf{y}$, their commutator is defined to be $[\bf{x},\bf{y}] \triangleq \bf{x}\bf{y}-\bf{y}\bf{x}$.

\item For a matrix $X=[x_{ij}]$ with  number or operator entries,  $X^{\top}=[x_{ji}]$ is the matrix transpose. Denote $X^{\#}=[x_{ij}^{\ast}]$, and $X^{\dagger}=(X^{\#})^{\top}$. For a vector $x$, we define $\breve{x}\triangleq \bigl[
\begin{smallmatrix}
x \\
x^{\#}
\end{smallmatrix}
\bigr]$.

\item Given two \textit{column} vectors of operators $\bf{X}$ and $\bf{Y}$ of the same length, their commutator is defined as
\beq
[\bf{X},\bf{Y}^\top] \triangleq ([\bf{X}_j,\bf{Y}_k] ) =\bf{X}\bf{Y}^\top- (\bf{Y}\bf{X}^\top)^\top.
\eeq
If $\bf{X}$ is a \textit{row} vector of operators of length $m$ and $\bf{Y}$ is a \textit{column} vector of operators of length $n$, their commutator is defined as
\begin{equation}\label{dec19-6}
[\mbf{X},\mbf{Y}] \triangleq \left(\begin{array}{@{}ccc@{}}                               [\mbf{x}_1,\mbf{y}_1] & \cdots & [\mbf{x}_m,\mbf{y}_1] \\
                               {\vdots} & \ddots & \vdots \\
                               {[\mbf{x}_1,\mbf{y}_n]} & \cdots & [\mbf{x}_m,\mbf{y}_n]
\end{array}
                           \right)_{n\times m}=(\mbf{X}^\top\mbf{Y}^\top)^\top-\mbf{Y}\mbf{X}.
\end{equation}


\item Let $J_{k} \triangleq \mathrm{diag}(I_k,-I_k)$. For a matrix $X\in
\mathbf{C}^{2k\times 2r}$, define its $\flat$-adjoint by $X^{\flat }
\triangleq J_{r}X^{\dagger}J_{k}$. The $\flat$-adjoint operation enjoys the following  properties:
\beq
(x_1 A + x_2 B)^{\flat}=x_1^{*}
A^{\flat} + x_2^{*} B^{\flat}, \ \ (AB)^{\flat}=B^{\flat} A^{\flat}, \ \
(A^{\flat})^{\flat}=A,
\eeq
where $x_1,x_2\in \mathbf{C}$.

\item Given two matrices $U$, $V\in \mathbf{C}^{k\times r}$, define their
 \emph{doubled-up} \cite{GJN10} as  $\Delta
(U,V) \triangleq
\bigl[
\begin{smallmatrix}
U & V \\
V^{\#} & U^{\#}
\end{smallmatrix}
\bigr]$.  The set
of doubled-up matrices is closed under addition, multiplication and $\flat$ adjoint operation.

\item A matrix $T \in \mathbf{C}^{2k\times 2k}$ is called \emph{Bogoliubov}
if it is doubled-up and satisfies $TT^{\flat}=T^{\flat}T=I_{2k}$. The set of Bogoliubov matrices
forms a complex non-compact Lie group known as the Bogoliubov group.

\item Let $\mathbb{J}_{k} \triangleq \bigl[
\begin{smallmatrix}
0_{k} & I_k \\
-I_k & 0_{k}
\end{smallmatrix}
\bigr]$. For a matrix $X\in \mathbf{C}^{2k\times 2r}$, define its $\sharp$-
\emph{adjoint} $X^{\sharp}$ by $X^{\sharp} \triangleq -\mathbb{J}_{r}X^{\dagger}
\mathbb{J}_{k}$. The $\sharp$-\emph{adjoint} satisfies properties similar to
the usual adjoint, namely
\beq
(x_1 A + x_2 B)^{\sharp}=x_1^{*} A^{\sharp} + x_2^{*}
B^{\sharp}, \ \ (AB)^{\sharp}=B^{\sharp} A^{\sharp},  \ \ (A^{\sharp})^{
\sharp}=A.
\eeq

\item A matrix $\mathbb{S} \in \mathbf{C}^{2k\times 2k}$ is called \emph{symplectic},
if $\mathbb{S}\mathbb{S}^{\sharp}=\mathbb{S}^{\sharp}\mathbb{S}=I_{2k}$. 

\end{itemize}

The remainder of this paper is organized as follows. Section \ref{preliminaries} reviews the essential preliminaries of linear quantum systems, including the quantum stochastic differential equation, Heisenberg-picture dynamics, and input–output relations. Section \ref{sec:qndbae} establishes the framework for realizing BAE measurements. It first provides sufficient conditions for bilateral BAE measurements realizations under purely imaginary inherent system Hamiltonian and real or imaginary coupling operators, then extends to unilateral BAE measurement realizations when the real parts of the system Hamiltonian are equal or opposite. A special case involving Michelson interferometry is presented, followed by a coherent feedback control scheme that engineers BAE measurements for non-compliant systems. Section \ref{QND INT} investigates QND interaction. It begins with the single-input–single-output (SISO) case and progressively treats multi-input–multi-output (MIMO) systems in the annihilation–creation, quadrature, and Kalman canonical forms. The QND interaction condition is shown to simultaneously guarantee BAE measurements and promote the coupling operator to a QND observable. Section \ref{Conclu} concludes the paper and outlines future directions.

\section{Preliminaries}\label{preliminaries}

The time evolution of a linear quantum system is governed by a unitary operator $\mbf{U}(t,t_0)$, which satisfies the following quantum stochastic differential equation (QSDE)

\begin{equation}\label{dU}
d\mbf{U}(t,t_0)=\left[\left(-\imath \mbf{H}-\frac{1}{2}\mbf{L}^\dagger \mbf{L}\right)dt-\mbf{L}^\dagger S d\mbf{B}_{\rm in}(t)+d\mbf{B}_{\rm in}^\dagger(t)\mbf{L}+{\rm Tr}\left[(S-I)d\mbf{\Lambda}^\top(t)\right]\right]\mbf{U}(t,t_0),    
\end{equation}
where $\mbf{U}(t_0,t_0)=I$. $S$ denotes the scattering matrix, $\mbf{L}$ is the coupling operator with $\mbf{L}=\left[\begin{array}{cc}
C_- & C_+ 
\end{array}\right]\breve{\mbf{a}}$, and $\mbf{H}$ is the system Hamiltonian $\mbf{H}=\frac{1}{2}\breve{\mbf{a}}^\dagger\Omega\breve{\mbf{a}}$, where $\Omega=\Delta(\Omega_-,\Omega_+)$. The integrated input annihilation, creation, and gauge processes are
\begin{equation}\begin{aligned}
\mbf{B}_{\rm in}(t)=\int_{t_0}^t \mbf{b}_{\rm in}(s)ds, ~~ \mbf{B}_{\rm in}^\ast(t)=\int_{t_0}^t \mbf{b}_{\rm in}^\ast(s)ds, ~~ \mbf{\Lambda}(t)=\int_{t_0}^t \mbf{b}_{\rm in}^\ast(s) \mbf{b}_{\rm in}(s)ds,
\end{aligned}\end{equation}
respectively, and satisfy the following commutation relation
\begin{equation*}
\left[\mbf{b}_j(t),\mbf{b}^\ast_k(s)\right]=\delta_{jk}\delta(t-s).
\end{equation*}
In the Heisenberg picture, denote a system operator $\mbf{X}(t)=\mbf{U}^\ast(t,t_0)(\mbf{X}\otimes I_{\rm field})\mbf{U}(t,t_0)$, and its dynamical evolution can be described as

\begin{equation}\label{dX}\begin{aligned}
d\mbf{X}(t)=&\mathcal{L}(\mbf{X}(t))dt+[\mbf{L}^\dagger(t),\mbf{X}(t)]Sd\mbf{B}_{\rm in}(t) \\
&+d\mbf{B}_{\rm in}^\dagger(t) S^\dagger[\mbf{X}(t),\mbf{L}(t)]+{\rm Tr}\left[(S^\dagger \mbf{X}(t)S-\mbf{X}(t))d\mbf{\Lambda}^\top(t)\right],    
\end{aligned}\end{equation}
where the superoperator 
\begin{equation}
\mathcal{L}(\mbf{X}(t))=-\imath[\mbf{X}(t),\mbf{H}(t)]+\frac{1}{2}\mbf{L}^\dagger(t)[\mbf{X}(t),\mbf{L}(t)]+\frac{1}{2}[\mbf{L}^\dagger(t),\mbf{X}(t)]\mbf{L}(t).    
\end{equation}
The input-output relation is given by
\begin{equation}\label{eq:io}
d\mbf{B}_{\rm out}(t)=\mbf{L}(t)dt+Sd\mbf{B}_{\rm in}(t).
\end{equation}

In the Schr{\"o}dinger picture, the conditioned density operator $\rho_c$ is governed by the following quantum stochastic master equation (QSME)
\begin{equation}\label{Jan12-2}
d\rho_c(t)=\mathcal{L}^\star(\rho_c(t))dt+\{\mbf{L}^\top\rho_c(t)+\rho_c(t)\mbf{L}^\dagger-{\rm Tr}[\rho_c(t)(\mbf{L}^\top+\mbf{L}^\dagger)]\rho_c(t)\}d\nu(t),
\end{equation}
where the superoperator $\mathcal{L}^\star(\rho_c(t))$ is defined by
\begin{equation}
\mathcal{L}^\star(\rho_c(t))=-\imath[\mbf{H},\rho_c(t)]+\mbf{L}^\top \rho_c(t) \mbf{L}^\#-\frac{1}{2}\mbf{L}^\dagger \mbf{L}\rho_c(t)-\frac{1}{2}\rho_c(t)\mbf{L}^\dagger \mbf{L},
\end{equation}
and the innovation process 
\begin{equation}
d\nu(t)=d\mbf{Q}_{\rm out}(t)-{\rm Tr}
[\rho_c(t)(\mbf{L}+\mbf{L}^\#)]dt,
\end{equation}
where $d\nu(t)d\nu^\top(t)=dt$, and $\mbf{Q}_{\rm out}(t)=\frac{1}{\sqrt{2}}\left(\mbf{B}_{\rm out}(t)+\mbf{B}_{\rm out}^\#(t)\right)$ being assumed to be observed continuously.

\section{Realization of BAE measurements}\label{sec:qndbae}

A linear quantum system with the triple $(S,\mbf{L},\mbf{H})$ language \cite{GJ09} can be described in the annihilation-creation form
\begin{equation}\label{system:ani}\begin{aligned}
\dot{\breve{\mbf{a}}}&=\mathcal{A}\breve{\mbf{a}}+\mathcal{B}\breve{\mbf{b}}_{\rm in}, \\
\breve{\mbf{b}}_{\rm out}&=\mathcal{C}\breve{\mbf{a}}+\mathcal{D}\breve{\mbf{b}}_{\rm in},
\end{aligned}\end{equation}
where the system matrices 
\begin{equation}\begin{aligned}
&\mathcal{C}=\Delta(C_-,C_+), ~~ \mathcal{D}=\Delta(S,0), \\
&\mathcal{B}=-\mathcal{C}^\flat\mathcal{D}, ~~ \mathcal{A}=-\imath J_n\Omega-\frac{1}{2}\mathcal{C}^\flat\mathcal{C}.
\end{aligned}\end{equation}

Alternatively, the linear quantum system \eqref{system:ani} is equivalent to the following
(real) quadrature operator representation

\begin{equation}\label{system:qua}\begin{aligned}
\dot{\mbf{x}}&=\mathbb{A}\mbf{x}+\mathbb{B}\mbf{u}, \\
\mbf{y}&=\mathbb{C}\mbf{x}+\mathbb{D}\mbf{u},
\end{aligned}\end{equation}
where system state, input, and output are
\begin{equation}\begin{aligned}
\mbf{x}=\left[\begin{array}{c}
\mbf{q} \\
\mbf{p}
\end{array}\right], ~~ \mbf{u}=\left[\begin{array}{c}
\mbf{q}_{\rm in} \\
\mbf{p}_{\rm in}
\end{array}\right], ~~ \mbf{y}=\left[\begin{array}{c}
\mbf{q}_{\rm out} \\
\mbf{p}_{\rm out}
\end{array}\right],
\end{aligned}\end{equation}
respectively, and
\begin{equation}\la{eq:mar24_ABCD}
\begin{aligned}
&\mathbb{D}=V_m \mathcal{D} V_m^\dagger =\left[
                                          \begin{array}{cc}
                                            \mathrm{Re}(S) & -\mathrm{Im}(S) \\
                                            \mathrm{Im}(S) & \mathrm{Re}(S) \\
                                          \end{array}
                                        \right], ~~~~ \mathbb{C}=V_m \mathcal{C} V_n^\dagger =\left[
                                                                             \begin{array}{cc}
                                                                               \mathrm{Re}(C_-+C_+) & -\mathrm{Im}(C_--C_+) \\
                                                                               \mathrm{Im}(C_-+C_+) & \mathrm{Re}(C_--C_+) \\
                                                                             \end{array}
                                                                           \right], \\
&\mathbb{B}=V_n \mathcal{B} V_m^\dagger = -\left[
                                            \begin{array}{cc}
                                              \mathrm{Re}(C_-^\dagger-C_+^\dagger) & -\mathrm{Im}(C_-^\dagger-C_+^\dagger) \\
                                              \mathrm{Im}(C_-^\dagger+C_+^\dagger) & \mathrm{Re}(C_-^\dagger+C_+^\dagger) \\
                                            \end{array}
                                          \right]\mathbb{D}, ~~~~ \mathbb{A} = V_n \mathcal{A} V_n^\dagger = \mathbb{J}_n\mathbb{H}-\frac{1}{2}\mathbb{C}^\sharp \mathbb{C},
\end{aligned}
\end{equation}
with
\begin{equation} \la{eq:mar24_JH}
\mathbb{J}_n\mathbb{H}=\mathbb{J}_n V_n \Omega V_n^\dagger=\left[
                                                             \begin{array}{cc}
                                                               \mathrm{Im}(\Omega_-+\Omega_+) & \mathrm{Re}(\Omega_--\Omega_+) \\
                                                               -\mathrm{Re}(\Omega_-+\Omega_+) & \mathrm{Im}(\Omega_--\Omega_+) \\
                                                             \end{array}
                                                           \right],
\end{equation}
and
\begin{equation} \label{eq:apr_C_sharp_C}
\mathbb{C}^\sharp\mathbb{C}=\left[
                                                    \begin{array}{cc}
                                                      \mathrm{Re}(C_-^\dagger C_- - C_+^\dagger C_+ +C_-^\dagger C_+ - C_+^\top C_-^\#) & -\mathrm{Im}(C_-^\dagger C_- + C_+^\dagger C_+ +C_-^\dagger C_+ - C_+^\top C_-^\#) \\
                                                      \mathrm{Im}(C_-^\dagger C_- + C_+^\dagger C_+ -C_-^\dagger C_+ + C_+^\top C_-^\#) & \mathrm{Re}(C_-^\dagger C_- - C_+^\dagger C_+ -C_-^\dagger C_+ + C_+^\top C_-^\#) \\
                                                    \end{array}
                                                  \right].
\end{equation}

In the following subsections, we respectively address the sufficient conditions under which bilateral and unilateral quantum BAE measurements can be implemented.


\subsection{Bilateral quantum BAE measurements}



Firstly, we assume that $\Omega$ is purely imaginary and $\mathcal{C}$ is real. The scattering operator $S$ is real, e.g., $S=I$ for simplicity.
In this case, Eq. \eqref{eq:mar24_JH} reduces to
\begin{equation} \label{eq:mar24_JH5}
\mathbb{J}_n\mathbb{H}=-\imath \left[
                         \begin{array}{cc}
                           \Omega_-+\Omega_+ & 0 \\
                           0 & \Omega_--\Omega_+ \\
                         \end{array}
                       \right].
\end{equation}
Similarly, by Eq. \eqref{eq:mar24_ABCD},
\begin{equation}\la{eq:mar24_BC}
\begin{aligned}
& \mathbb{C}=V_m \mathcal{C} V_n^\dagger =\left[
                                                                             \begin{array}{cc}
                                                                              C_-+C_+ & 0 \\
                                                                               0 & C_--C_+ \\
                                                                             \end{array}
                                                                           \right], \\
&\mathbb{B}=V_n \mathcal{B} V_m^\dagger =-\left[
                                            \begin{array}{cc}
                                              (C_--C_+)^\top & 0 \\
                                            0& (C_-+C_+)^\top) \\
                                            \end{array}
                                          \right],
\end{aligned}
\end{equation}
and accordingly  Eq. \eqref{eq:apr_C_sharp_C} reduces to
\begin{equation}\la{eq:mar24_A}\begin{aligned}
\mathbb{C}^\sharp\mathbb{C}
 =\left[
                                                    \begin{array}{cc}
                                                     (C_--C_+)^\top (C_-+C_+) & 0 \\
                                                      0&(C_-+C_+)^\top  (C_--C_+)  \\
                                                    \end{array}
                                                  \right].
     \end{aligned}
\end{equation}
According to Eqs. \eqref{eq:mar24_JH5} and \eqref{eq:mar24_A}, we have
\beq\la{eqLmar34_A5}
\mathbb{A} = -\imath \left[
                         \begin{array}{cc}
                           \Omega_-+\Omega_+ & 0 \\
                           0 & \Omega_--\Omega_+ \\
                         \end{array}
                       \right]-\frac{1}{2}\left[
                                                    \begin{array}{cc}
                                                     (C_--C_+)^\top (C_-+C_+) & 0 \\
                                                      0&(C_-+C_+)^\top  (C_--C_+)  \\
                                                    \end{array}
                                                  \right].
\eeq

Denote the matrices
\beq\label{eq:jun29_C+-}
\mathbb{C}_q \triangleq C_-+C_+,  \ \ \ 
\mathbb{C}_p \triangleq C_--C_+.
\eeq
The following result can be derived directly.

\bprop\label{prop:BAE}
If $\Omega$ is purely imaginary, both $S$ and $\mathcal{C}$ are real, then the transfer function is of the form
\beq\label{Mar15-1}
\begin{aligned}
&\mathbb{G}[s] =  \left[ 
\bey{cc}
\mathbb{G}_q[s] & 0 \\
0  & \mathbb{G}_p[s] 
\eey
\right],
\end{aligned}
\eeq
where
\beq
\begin{aligned}
\mathbb{G}_q[s]=
S-\mathbb{C}_q [sI+\imath(\Omega_-+\Omega_+)+\frac1{2}\mathbb{C}_p^\top \mathbb{C}_q]^{-1} \mathbb{C}_p^\top, \\
\mathbb{G}_p[s]=
S-\mathbb{C}_p [sI+\imath(\Omega_--\Omega_+)+\frac1{2}\mathbb{C}_q^\top \mathbb{C}_p]^{-1} \mathbb{C}_q^\top.
\end{aligned}
\eeq
\eprop

\begin{remark}
In this case, BAE measurements of $\mbf{q}_{\rm out}$ with respect to $\mbf{p}_{\rm in}$ and $\mbf{p}_{\rm out}$ with respect to $\mbf{q}_{\rm in}$  are realized \cite[Theorem 4.1]{ZPL20}. Moreover, noticing that the system can be non-passive as $C_+$ and $\Omega_+$ can be nonzero in this case. The non-degenerate parametric amplifier (NDPA) (after rotation)  in \cite{SY21}, the DPA,  and the model studied in \cite{BQND24} (ignoring the perturbation term) belong to this special case.
\end{remark}

On the other hand, if $\mathcal{C}$ is purely imaginary in Proposition \ref{prop:BAE}, the system matrices $\mathbb{C}$ and $\mathbb{B}$ can be described by the following block off-diagonal forms

\begin{equation}\begin{aligned}
&\mathbb{C}=\left[\begin{array}{cc}
0 & -\mathrm{Im}(C_--C_+) \\
\mathrm{Im}(C_-+C_+) & 0
\end{array}\right], \\
&\mathbb{B}=\left[\begin{array}{cc}
0 & \mathrm{Im}(C_-^\dagger-C_+^\dagger) \\
-\mathrm{Im}(C_-^\dagger+C_+^\dagger) & 0
\end{array}\right].
\end{aligned}\end{equation}
Then
\begin{equation}\begin{aligned}
\mathbb{C}^\sharp\mathbb{C}=
\left[\begin{array}{cc}
(C_-^\dagger - C_+^\dagger)(C_- + C_+) & 0 \\
0 & (C_-^\dagger + C_+^\dagger)(C_- - C_+)
\end{array}\right],
\end{aligned}\end{equation} 
which means that the system matrix $\mathbb{A}$ is block diagonal. Consequently, the transfer function in the real quadrature form is also block diagonal, which indicates that in this case the system position $\mbf{q}$ and momentum $\mbf{p}$ are separable. Specifically, the dynamical evolution of this system is expressed by

\begin{equation}\label{Jan9-1}\begin{aligned}
\dot{\mbf{q}} =& \left(-\imath(\Omega_-+\Omega_+)-
\frac{1}{2}(C_- - C_+)^\dagger(C_-+C_+)\right)\mbf{q}+ \mathrm{Im}(C_--C_+)^\dagger\mbf{p}_{\rm in}, \\
\dot{\mbf{p}} =& \left(-\imath(\Omega_--\Omega_+)-
\frac{1}{2}(C_- + C_+)^\dagger(C_--C_+)\right)
\mbf{p} -\mathrm{Im}(C_-+C_+)^\dagger\mbf{q}_{\rm in}, \\
\mbf{q}_{\mathrm{out}}=&\; -\mathrm{Im}(C_--C_+)\mbf{p} +\mbf{q}_{\rm in}, \\
\mbf{p}_{\mathrm{out}} =&\; \mathrm{Im}(C_-+C_+)\mbf{q} + \mbf{p}_{\rm in},
\end{aligned}\end{equation}

where quantum BAE measurements of $\mbf{q}_{\rm out}$ with respect to $\mbf{p}_{\rm in}$ and $\mbf{p}_{\rm out}$ with respect to $\mbf{q}_{\rm in}$ are realized \cite[Theorem 4.1]{ZPL20}. Moreover, noticing that the system can be non-passive as $C_+$ and $\Omega_+$ can be nonzero in this case. When $C_-=C_+$, $\mbf{q}$ is a QND variable if the subsystem $(\imath(\Omega_-+\Omega_+),0,C_-)$ is observable; when $C_-=-C_+$, $\mbf{p}$ is a QND variable if the subsystem $(\imath(\Omega_--\Omega_+),0,C_-)$ is observable \cite[Remark 4.9]{ZGPG18}.

The following lemma summarizes the main results of the preceding two special cases.

\begin{lemma}\label{lem:BAE1}
When $S$ is real, $\Omega$ is purely imaginary and $\mathcal{C}$ is real or purely imaginary, quantum BAE measurements of $\mbf{q}_{\rm out}$ with respect to $\mbf{p}_{\rm in}$ and $\mbf{p}_{\rm out}$ with respect to $\mbf{q}_{\rm in}$ are realized.
If further 
\begin{itemize}
\item $C_-=C_+$ ($\mbf{q}$ is coupled with the environment), then $\mbf{q}$ is a QND variable if the subsystem $(\imath(\Omega_-+\Omega_+),0,C_-)$ is observable;

\item $C_-=-C_+$ ($\mbf{p}$ is coupled with the environment), then $\mbf{p}$ is a QND variable if the subsystem $(\imath(\Omega_--\Omega_+),0,C_-)$ is observable.
\end{itemize}
\end{lemma}

Secondly, if the scattering operator $S$ is purely imaginary, e.g., $S=-\imath I$, then the system matrix $\mathbb{D}=\left[\begin{array}{cc}
0 & I \\
-I & 0
\end{array}\right]$ is block off-diagonal. If further $\mathcal{C}$ is assumed to be real, then $\mathbb{C}=\left[\begin{array}{cc}
C_-+C_+ & 0 \\
0 & C_--C_+
\end{array}\right]$, and $\mathbb{B}=\left[\begin{array}{cc}
0 & -(C_--C_+)^\top \\
(C_-+C_+)^\top & 0
\end{array}\right]$. In this case, the transfer function $\mathbb{G}[s]$ is also block off-diagonal, and thus BAE measurements of $\mbf{q}_{\rm out}$ with respect to $\mbf{q}_{\rm in}$ and $\mbf{p}_{\rm out}$ with respect to $\mbf{p}_{\rm in}$  are realized. Analogous bilateral quantum BAE measurements can be realized when $\mathcal{C}$ is purely imaginary. We have the following result.

\begin{lemma}\label{lem:BAE2}
When both $S$ and $\Omega$ are purely imaginary, $\mathcal{C}$ is real or purely imaginary, quantum BAE measurements of  $\mbf{q}_{\rm out}$ with respect to $\mbf{q}_{\rm in}$ and $\mbf{p}_{\rm out}$ with respect to $\mbf{p}_{\rm in}$  are realized.
If further 
\begin{itemize}
\item $C_-=C_+$ ($\mbf{q}$ is coupled with the environment), then $\mbf{q}$ is a QND variable if the subsystem $(\imath(\Omega_-+\Omega_+),0,C_-)$ is observable;

\item $C_-=-C_+$ ($\mbf{p}$ is coupled with the environment), then $\mbf{p}$ is a QND variable if the subsystem $(\imath(\Omega_--\Omega_+),0,C_-)$ is observable.
\end{itemize}
\end{lemma}


\subsection{Unilateral quantum BAE measurement}

In this section, we continue to discuss a more general case where $\Omega_-$ and $\Omega_+$ are not purely imaginary, and their real parts are either equal or opposite in sign.

\subsubsection{$\Omega_-$ and $\Omega_+$ share the same real part}

In this case, $\mathbb{J}_n\mathbb{H}$ is block lower triangular, and so is system matrix $\mathbb{A}$. We address the following four special cases.

\begin{itemize}

    \item Both $S$ and $\mathcal{C}$ are real matrices, the transfer function $\mathbb{G}[s]$ is block lower triangular, and thus BAE measurements of $\mbf{q}_{\rm out}$ with respect to $\mbf{p}_{\rm in}$ is realized. 

    \item $S$ is real and $\mathcal{C}$ is purely imaginary, the transfer function $\mathbb{G}[s]$ is block upper triangular, and thus BAE measurements of $\mbf{p}_{\rm out}$ with respect to $\mbf{q}_{\rm in}$ is realized.
    
    \item $S$ is purely imaginary and $\mathcal{C}$ is real, the top-left 1-by-1 block of the transfer function $\mathbb{G}[s]$ is zero, and thus BAE measurements of $\mbf{q}_{\rm out}$ with respect to $\mbf{q}_{\rm in}$ is realized.
    
    \item Both $S$ and $\mathcal{C}$ are purely imaginary matrices, the bottom-right 2-by-2 block of the transfer function $\mathbb{G}[s]$ is zero, and thus BAE measurements of $\mbf{p}_{\rm out}$ with respect to $\mbf{p}_{\rm in}$ is realized.

\end{itemize}

\subsubsection{$\Omega_-$ and $\Omega_+$ have opposite real parts}

In this case, $\mathbb{J}_n\mathbb{H}$ is block upper triangular, and so is system matrix $\mathbb{A}$. We address the following four special cases.

\begin{itemize}

    \item Both $S$ and $\mathcal{C}$ are real matrices, the transfer function $\mathbb{G}[s]$ is block upper triangular, and thus BAE measurements of $\mbf{p}_{\rm out}$ with respect to $\mbf{q}_{\rm in}$ is realized. 

    \item $S$ is real and $\mathcal{C}$ is purely imaginary, the transfer function $\mathbb{G}[s]$ is block lower triangular, and thus BAE measurements of $\mbf{q}_{\rm out}$ with respect to $\mbf{p}_{\rm in}$ is realized.
    
    \item $S$ is purely imaginary and $\mathcal{C}$ is real, the bottom-right 2-by-2 block of the transfer function $\mathbb{G}[s]$ is zero, and thus BAE measurements of $\mbf{p}_{\rm out}$ with respect to $\mbf{p}_{\rm in}$ is realized.
    
    \item Both $S$ and $\mathcal{C}$ are purely imaginary matrices, the top-left 1-by-1 block of the transfer function $\mathbb{G}[s]$ is zero, and thus BAE measurements of $\mbf{q}_{\rm out}$ with respect to $\mbf{q}_{\rm in}$ is realized.

\end{itemize}

\subsection{A special case}

In this section, we no longer impose constraints on the system Hamiltonian $\Omega$. Instead, we focus on the case where the coupling operator matrix $\mathcal{C}$ is purely imaginary and satisfies $C_-=\pm C_+$, i.e., $\mbf{q}$-coupling or $\mbf{p}$-coupling cases. 

\begin{corollary}\label{cor:BAE}
When $S$ is real, $\mathcal{C}$ is purely imaginary, and $C_-=C_+$,
quantum BAE measurement of $\mbf{q}_{\rm out}$ with respect to $\mbf{p}_{\rm in}$ is realized.
\end{corollary}

The following example given in \cite[Fig. 3(c)]{NY14}, \cite[Eq. (35)]{ZPL20} demonstrates Corollary \ref{cor:BAE}.

\begin{example}
Consider the Michelson’s interferometer, which is one of the simplest devices for gravitational wave detection, with the system Hamiltonian 
\begin{equation}
\Omega_-=\frac{1}{2}\left[\begin{array}{cc}
m \omega_m^2+\frac{1}{m} & 0 \\
0 & m \omega_m^2+\frac{1}{m}
\end{array}\right], \ \ \Omega_+=\frac{1}{2}\left[\begin{array}{cc}
m \omega_m^2-\frac{1}{m} & 0 \\
0 & m \omega_m^2-\frac{1}{m}
\end{array}\right],
\end{equation}
where $m$ and $\omega_m$ denote the mass of the mechanical
oscillators and the resonant frequency, respectively. The coupling operator matrix $\mathcal{C}$ is purely imaginary and 
\begin{equation}
C_-=C_+=\frac{\sqrt{\lambda}}{2}\left[\begin{array}{cc}
\imath & \imath \\
\imath & -\imath
\end{array}\right],    
\end{equation}
where $\lambda$ is the coupling strength between the input field and the mechanical oscillators. The scattering matrix $S$ is chosen to be identity, i.e., $S= I$. As a result, the transfer function of the Michelson’s interferometer $\mathbb{G}[s]$ can be calculated as a block lower triangular matrix, which indicates that quantum BAE measurement of $\mbf{q}_{\rm out}$ with respect to $\mbf{p}_{\rm in}$ is realized.
\end{example}

The special case of $\mbf{p}$-coupling results in a block upper triangular system transfer function matrix, which is summarized by the following corollary.

\begin{corollary}
When $S$ is real, $\mathcal{C}$ is purely imaginary, and $C_-=-C_+$,
quantum BAE measurement of $\mbf{p}_{\rm out}$ with respect to $\mbf{q}_{\rm in}$ is realized.    
\end{corollary}


\subsection{Realization of quantum BAE measurements via coherent feedback control}

As discussed in the previous section, when the system Hamiltonian $\Omega$ is purely imaginary and the coupling operator matrix $\mathcal{C}$ is real or purely imaginary, quantum BAE measurements can be realized. However, if the sufficient conditions in Lemmas \ref{lem:BAE1}-\ref{lem:BAE2} are not applied to the original system, we can cascade a beamsplitter and thus form a coherent feedback control network to realize quantum BAE measurements. Similar optical system structures of the coherent feedback control for squeezing enhancement \cite[Fig. 1(b)]{IYYYF12} and noise reduction \cite[Fig. 1]{WZO2026} have been demonstrated experimentally. In this section, the original system $\boldsymbol{G}$, consisting of $n$ harmonic oscillators, $\mbf{a}=\left[\begin{array}{cccc}
\mbf{a}_1 & \mbf{a}_2 & \cdots & \mbf{a}_n
\end{array}\right]^\top$, is driven by $m=m_1+m_2$ input channels, whose $m_1$ input channels are described by $u_1$ in Fig. \ref{fig:system}, and the other $m_2$ input channels are represented by $u_2$. The system $\boldsymbol{G}$ is characterized by the following $(S,\mbf{L},\mbf{H})$ parameters
\begin{equation}\begin{aligned}
S_G=\left[\begin{array}{cc}
S_{11} & S_{12} \\
S_{21} & S_{22}
\end{array}\right], \ \mbf{L}_G=\left[\begin{array}{c}
\mbf{L}_1 \\
\mbf{L}_2 
\end{array}\right], \ \Omega_G=\Delta(\Omega_-,\Omega_+),    
\end{aligned}\end{equation}
where the coupling operators of the $m_1$ and $m_2$ channels are 
\begin{equation}\begin{aligned}
\mbf{L}_1=k_{11}\mbf{a}+k_{12}\mbf{a}^\#, \\
\mbf{L}_2=k_{21}\mbf{a}+k_{22}\mbf{a}^\#,
\end{aligned}\end{equation}
with the adjustable coupling strengths $k_{11}$, $k_{12}\in \mathbf{C}^{m_1 \times n}$, $k_{21}$, $k_{22}\in \mathbf{C}^{m_2 \times n}$.


\begin{figure}
    \centering
    \includegraphics[width=0.5\linewidth]{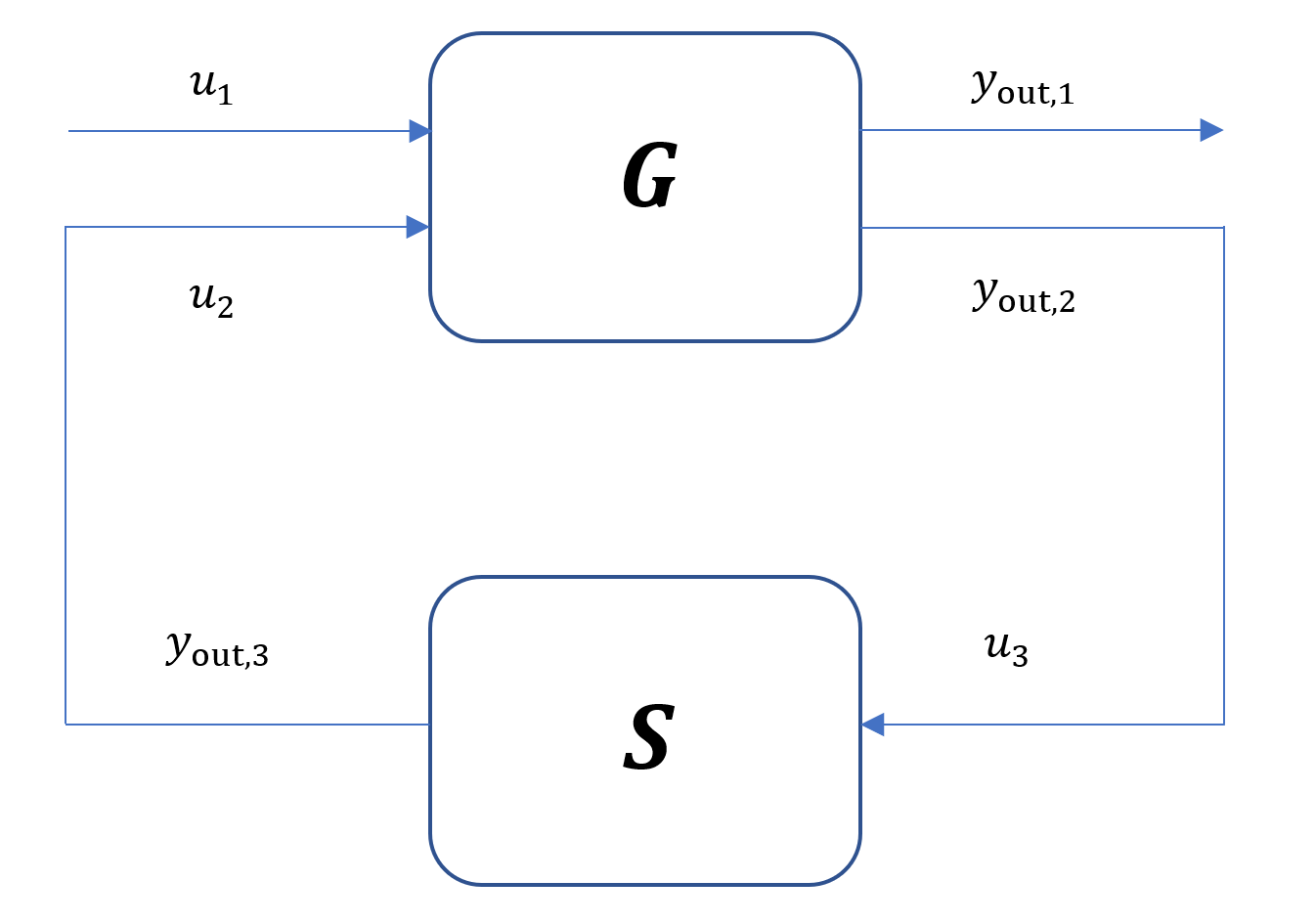}
    \caption{Coherent feedback network.}
    \label{fig:system}
\end{figure}

In contrast to the type-2 coherent feedback control scheme proposed in \cite[Fig. 9]{NY14}, we remove the controller and only use a beamsplitter $\boldsymbol{S}$ to form a coherent feedback network, see Fig. \ref{fig:system}. The feedback control is performed by modulating the corresponding $m_2$ output channels, i.e., $y_{\rm out,2}$, which is also the input $u_3$ of the beamsplitter $\boldsymbol{S}$, and the output $y_{\rm out,3}$ of the beamsplitter $\boldsymbol{S}$ is the input $u_2$ of the system $\boldsymbol{G}$. Without loss of generality, the scattering operator $S_b$ of the beamsplitter is chosen to be $I$ or $\imath I$. The total $(S,\mbf{L},\mbf{H})$ parameters of the coherent feedback network in Fig. \ref{fig:system} can be calculated as
\begin{equation}\begin{aligned}
S^{\rm red}&=S_{11}+S_{12}S_b(I-S_{22}S_b)^{-1} S_{21}, \\
\mbf{L}^{\rm red}&=\left[\begin{array}{cc}
\bar{C}_- & \bar{C}_+  
\end{array}\right]\left[\begin{array}{c}
\mbf{a} \\
\mbf{a}^\#      
\end{array}\right]
\end{aligned}    
\end{equation}
where $\bar{C}_-=k_{11}+S_{12}S_b(I-S_{22}S_b)^{-1}k_{21}$, $\bar{C}_+=k_{12}+S_{12}S_b(I-S_{22}S_b)^{-1}k_{22}$, and the reduced system Hamiltonian 
\begin{equation}
\bar{\Omega}=\Delta(\bar{\Omega}_-,\bar{\Omega}_+),    
\end{equation}
where 
\begin{equation}
\bar{\Omega}_-=\Omega_--\imath(k_{11}^\dagger S_b k_{21}-k_{21}^\dagger S_b^\dagger k_{11}), \ \ \bar{\Omega}_+=\Omega_+-\imath(k_{11}^\dagger S_b k_{22}-k_{21}^\dagger S_b^\dagger k_{12}).
\end{equation}

We have the following sufficient condition to realize BAE measurements of the coherent feedback network in Fig. \ref{fig:system}.
\begin{theorem}\label{thm:BAE}
Appropriate selection of the coupling strength parameters $k_{ij}$, $i,j=1,2$, to render the reduced Hamiltonian $\bar{\Omega}$ purely imaginary enables bilateral BAE measurements in the quantum coherent feedback network.  
\end{theorem}


\begin{remark}
The selection of coupling strength parameters is primarily based on eliminating the real parts in the inherent system Hamiltonian $\Omega$, with the additional requirement that $\bar{C}=\Delta(\bar{C}_-,\bar{C}_+)$ be either real or purely imaginary.
\end{remark}

The following numerical example validates Theorem \ref{thm:BAE}.

\begin{example}
Assume that the original system $\boldsymbol{G}$ is with the following system Hamiltonian
\begin{equation}
\Omega_-=\left[\begin{array}{cc}
2 & 3+2\imath \\
3-2\imath & 4
\end{array}\right], \ \Omega_+=\left[\begin{array}{cc}
2 & 3-\imath \\ 
3-\imath & 5
\end{array}\right].    
\end{equation}    
Clearly, the system Hamiltonian $\Omega$ does not satisfy the sufficient condition given in Proposition \ref{prop:BAE}. The coupling strengths of the upper $m_1$ channels in Fig. \ref{fig:system} are set as follows
\begin{equation}\begin{aligned}
k_{11}=\left[\begin{array}{cc}
1 & 1+\imath 
\end{array}\right], \ \ k_{12}=\left[\begin{array}{cc}
1 & 2-\imath 
\end{array}\right],
\end{aligned}\end{equation}
while the coupling strengths of the lower $m_2$ channels
\begin{equation}\begin{aligned}
k_{21}=\left[\begin{array}{cc}
1+\imath & 1+\imath 
\end{array}\right], \ \ k_{22}=\left[\begin{array}{cc}
1+\imath & 2+2\imath 
\end{array}\right].
\end{aligned}\end{equation}
Set the scattering matrix $S=-\imath I$, then the reduced Hamiltonian $\bar{\Omega}$ of the coherent feedback network can be calculated as
\begin{equation}
\bar{\Omega}_-=\left[\begin{array}{cc}
0 & -\imath \\
\imath & 0
\end{array}\right], \ \bar{\Omega}_+=\left[\begin{array}{cc}
0 & \imath \\ 
\imath & 3\imath
\end{array}\right],    
\end{equation}
which is purely imaginary. Moreover, the reduced coupling 
\begin{equation}
\left[\begin{array}{cc}
\bar{C}_- & \bar{C}_+ 
\end{array}\right]=\left[\begin{array}{cccc}
1 & 2 & 1 & 1
\end{array}\right]    
\end{equation}
By Lemma \ref{lem:BAE2}, bilateral quantum BAE measurements of $\mbf{q}_{\rm out}$ with respect to $\mbf{q}_{\rm in}$ and $\mbf{p}_{\rm out}$ with respect to $\mbf{p}_{\rm in}$ in Fig. \ref{fig:system} are realized.
\end{example}

\section{QND interaction}\label{QND INT}

A QND interaction enables the repeated measurement of a quantum observable without perturbing its value or its eigenstate distribution. This is achieved by designing the system–probe coupling such that the measurement operator commutes with the system Hamiltonian. For example, consider a set of mutually commuting \emph{self-adjoint}  coupling operators $\mbf{L} = [\mbf{L}_1, \ldots , \mbf{L}_m]^\top$, which can be used as measurement operators for a continuous-time projection measurement process. $[\mbf{L},\mbf{H}]=0$ is assumed in the so-called QND interaction \cite[Sec. 4.2.4]{NY17}, and thus the time evolution of the measurement operator $\mbf{L}$ can be calculated by Eq. \eqref{dX} ($S=I$)
\begin{equation} \label{eq:L_sept16}
d\mbf{L}_j= (-\imath [\mbf{L}_j, \mbf{H}]+\frac{1}{2}([\mbf{L}^\dagger,\mbf{L}_j]\mbf{L})dt+
[\mbf{L}^\dagger,\mbf{L}_j]d\mbf{B}_{\rm in}=0,
\end{equation}
which means that the observable $\mbf{L}_j$ remains as its initial value and is independent of the measurement process.  Moreover the transfer matrix is the identity matrix, therefore quantum BAE measurements are realized.

\bmrk
According to Eq. \eqref{eq:L_sept16}, $\mbf{L}$ is set of mutually commuting observables, then they are QND variables. Moreover, as the transfer function is an identity matrix, therefore, quantum BAE measurements are realized.  Thus, if $\mbf{L}$ is a set of mutually commuting observables and $[\mbf{L},\mbf{H}]=0$, then QND and BAE are realized simultaneously. 
\emrk

Belavkin quantum filtering equation, which is a classical stochastic differential equation, is given by
\begin{equation}\label{dec19_1}\begin{aligned}
 d\pi_t(\mbf{L}_j)
=  \pi_t(\mathcal{L}(\mbf{L}_j))dt + [\pi_t(\mbf{L}_j\mbf{L}^\top+\mbf{L}^\dagger \mbf{L}_j) -\pi_t(\mbf{L}^\top+\mbf{L}^\dagger)\pi_t(\mbf{L}_j)]d\nu(t),
\end{aligned}\end{equation}
where $\nu(t)$ is Wiener process \cite[Sec. 6.1]{BvHJ07} and the superoperator 
\begin{equation}
\mathcal{L}(\mbf{L}_j)=-\imath[\mbf{L}_j,\mbf{H}]+\frac{1}{2}\mbf{L}^\dagger[\mbf{L}_j,\mbf{L}]+\frac{1}{2}[\mbf{L}^\dagger,\mbf{L}_j]\mbf{L} =0.    
\end{equation}
Here, we assume the initial conditional state has the finite first and second moments, $\pi_0(\mbf{L}_j) = \mathrm{Tr}(\rho_{\rm S}(0)\mbf{L}_j)$, where $\rho_{\rm S}(0)$ is the initial system density state. Thus, we have $d\pi_t(\mbf{L}_j)
=  [\pi_t(\mbf{L}_j\mbf{L}^\top+\mbf{L}^\dag \mbf{L}_j) -\pi_t(\mbf{L}_j)\pi_t(\mbf{L}^\top+\mbf{L}^\dag)]d\nu(t)$. Taking expectation with respect to the Wiener process $\nu(t)$, yields that  $\mathbb{E}[ d\pi_t(\mbf{L}_j)]=0$.
As a result,
\begin{equation}\label{eq:sept18_QND}
\mathbb{E}[\pi_t(\mbf{L}_j)] = \mathbb{E}[\pi_0(\mbf{L}_j)]  = {\rm Tr}(\rho_S(0)\mbf{L}_j).    
\end{equation}
Moreover, $\mathbb{E}[\pi_t(\mbf{L}_j^2)] = \mathbb{E}[\pi_0(\mbf{L}_j^2)]  = {\rm Tr}(\rho_S(0)\mbf{L}_j^2)$.

\bmrk
According to Eq. \eqref{eq:sept18_QND}, under the QND interaction, measurements have no back action on the estimation of $\mbf{L}_j$. 
\emrk

In the following SISO case, we go to the Schr{\"o}dinger picture.

Apply the spectral decomposition to the observable $\mbf{L}$, which can be written as
\begin{equation}
\mbf{L}=\sum_{j=1}^k l_jP_{l_j},    
\end{equation}
where $l_j$, $j=1,2,\ldots,k$, are the eigenvalues of $\mbf{L}$, $P_{l_j}$ denotes the corresponding orthogonal projection operator, and resolves the identity $\sum_{j=1}^k P_{l_j}=I$.  Define  the probability
\begin{equation}\label{Jan12-1}
P_j(t)={\rm Tr}[\rho_c(t) P_{l_j}].
\end{equation}
Differentiate both sides of  \eqref{Jan12-1}, we have
\begin{equation}\label{Jan12-3}
dP_j(t)={\rm Tr}[d\rho_c(t)P_{l_j}].
\end{equation}

Assuming the system Hamiltonian $\mbf{H}=0$ and inserting \eqref{Jan12-2} into \eqref{Jan12-3}, implies that
\begin{equation}\label{Jan12-4}
dP_j(t)=2\left(l_j-\sum_{j^\prime=1}^k l_{j^\prime}P_{j^\prime}(t)\right)P_{j}d\nu(t).
\end{equation}
Take expectation to \eqref{Jan12-4}, we have $d\mathbb{E}(P_j(t))/dt\equiv 0$,  which indicates that $P_j(t)$ is a martingale and thus
\begin{equation}
\mathbb{E}(P_j(t))\equiv \mathbb{E} (P_j(0))={\rm Tr}[\rho_c(0)P_{l_j}].    
\end{equation}
Consequently, the conditioned system density operator $\rho_c(t)$ after measurement converges to a projection $P_{l_j}$ with {\it constant} probability $P_j(0)$. If the initial density operator $\rho(0)$ is exactly one of eigenspaces of $L$, then $\rho(0)$ is a steady state and remains unchanged during the measurement process. Similar discussions on the QND interaction measurement in the Schr{\"o}dinger picture can be found in \cite{TC2014}, \cite[Thm. 5.1]{LAM19}, \cite[Lemma 1]{CSR+20}.

\bmrk
Let
\beq
\mbf{H}  = \sum_j h_j P_{l_j}, \ \ \rho_c(t) = 
\sum_j \alpha_j(t) P_{l_j}.
\eeq
Then 
\beq
[\mbf{H},\rho_c(t)] = \sum_j h_j \alpha_j(t)P_{l_j}.
\eeq
Thus, in general, even $[\mbf{L},\mbf{H}]=0$, $[\mbf{H},\rho_c(t)] $ may not be zero, thus Eq. \eqref{Jan12-4} may not hold. For the passive case, it seems OK as you can rotate away $\mbf{H}$. But in general it is unclear.
\emrk

\subsection{The SISO case}\label{siso}

In this subsection, we firstly investigate a more general case, $[\mbf{L}
+\mbf{L}^\ast,\mbf{H}]=0$ or $[\mbf{L}
-\mbf{L}^\ast,\mbf{H}]=0$ in the single-input-single-output (SISO) quantum systems, which can be used to achieve the BAE measurements.

Since the input channel $m=1$, it can be easily verified that $C_-C_+^\top=C_+C_-^\top$. In this case, $[\mbf{L}+\mbf{L}^\ast,\mbf{H}]=0$ yields 
\begin{equation}
d(\mbf{L}+\mbf{L}^\ast)=\frac{1}{2}\left(\mbf{L}^\ast[\mbf{L}^\ast,\mbf{L}]+[\mbf{L}^\ast,\mbf{L}]\mbf{L}\right)dt+d\mbf{B}_{\mathrm{in}}^\ast[\mbf{L}^\ast,\mbf{L}]+[\mbf{L}^\ast,\mbf{L}]d\mbf{B}_{\mathrm{in}}.
\end{equation}
As $\mbf{L}=C_-\mbf{a}+C_+\mbf{a}^\#$, and 
\begin{equation}
[C_-^\#\mbf{a}^\#,C_-\mbf{a}]=-\sum_j |C_{-,j}|^2, ~~~~~ [C_+^\#\mbf{a},C_+\mbf{a}^\#]=\sum_k |C_{+,k}|^2,
\end{equation}
we have
\begin{equation} \label{eq:sep12_1}
d\left(\mbf{L}+\mbf{L}^\ast\right)=-\frac{g}{2}\left(\mbf{L}+\mbf{L}^\ast\right)dt-\sqrt{2}g d\mbf{Q}_{\mathrm{in}},
\end{equation}
where 
\beq\label{sepg}
g\triangleq\sum_{j=1}^n (|C_{-,j}|^2-|C_{+,j}|^2), ~~ \mbf{Q}_{\mathrm{in}}=\frac{1}{\sqrt{2}}\left(\mbf{B}_{\mathrm{in}}+\mbf{B}^\ast_{\mathrm{in}}\right),
\eeq
and
\beq
(\mbf{L+L^\ast})(t) = e^{-gt/2}(\mbf{L+L^\ast})(0)-\sqrt{2}g \int_0^t e^{-g(t-\tau)/2}d\mbf{Q}_{\mathrm{in}}(\tau),
\eeq
\beq\label{sep11-1}
\mbf{q}_{\rm out}(t) =\frac{1}{\sqrt{2}}(\mbf{L+L^\ast})(t) + \mbf{q}_{\mathrm{in}}(t).
\eeq
By Eq. \eqref{eq:sep12_1},
\beq \label{eq:sept16_q_out}
\mbf{q}_{\rm out}[s] = [1-g (s+\frac{g}{2})^{-1}] \mbf{q}_{\mathrm{in}}[s] = \frac{s-g/2}{s+g/2}\mbf{q}_{\mathrm{in}}[s]. 
\eeq
Thus, the BAE measurement of $\mbf{q}_{\rm out}$ with respect to $\mbf{p}_{\rm in}$ is realized.

On the other hand, $[\mbf{L}-\mbf{L}^\ast,\mbf{H}]=0$ implies that
\begin{equation}\label{Oct16-1}\begin{aligned}
d(\mbf{L}-\mbf{L}^\ast)=-\frac{g}{2}(\mbf{L}-\mbf{L}^\ast)dt-\imath\sqrt{2}g d\mbf{P}_{\mathrm{in}},
\end{aligned}\end{equation}
where $\mbf{P}_{\mathrm{in}}=\frac{-\imath}{\sqrt{2}}\left(\mbf{B}_{\mathrm{in}}-\mbf{B}^\ast_{\mathrm{in}}\right)$. Solving Eq. \eqref{Oct16-1}, yields that 
\beq
(\mbf{L-L^\ast})(t) = e^{-gt/2}(\mbf{L-L^\ast})(0)-\imath\sqrt{2}g \int_0^t e^{-g(t-\tau)/2} d\mbf{P}_{\mathrm{in}}(\tau).
\eeq
We have
\beq\label{sep11-2}
\mbf{p}_{\rm out}(t) =\frac{-\imath}{\sqrt{2}}(\mbf{L-L^\ast})(t) + \mbf{p}_{\mathrm{in}}(t), 
\eeq
and thus
\beq
\mbf{p}_{\rm out}[s] = [1-g (s+\frac{g}{2})^{-1}] \mbf{p}_{\mathrm{in}}[s] = \frac{s-g/2}{s+g/2}\mbf{p}_{\mathrm{in}}[s], 
\eeq
which realizes the BAE measurement of $\mbf{p}_{\rm out}$ with respect to $\mbf{q}_{\rm in}$. 

Consequently, quantum BAE measurement in the SISO case can be realized if $[\mbf{L}+\mbf{L}^\ast,\mbf{H}]=0$ or $[\imath(\mbf{L}-\mbf{L}^\ast),\mbf{H}]=0$.

\begin{remark}
If $\sum_j |C_{-,j}|^2=\sum_k |C_{+,k}|^2$, then $g=0$ in \eqref{sepg}, which yields that $d(\mbf{L}+\mbf{L}^\ast)=0$ and the coupling operator $(\mbf{L}+\mbf{L}^\ast)(t)=(\mbf{L}+\mbf{L}^\ast)(0)=\mbf{L}+\mbf{L}^\ast$ can be detected via the input-output relation \eqref{sep11-1}. Thus, this kind of interaction does not change the physical property of the coupling operator $\mbf{L}+\mbf{L}^\ast$,
which is called a QND interaction in \cite[Sec. 4.2.4]{NY17}. On the other hand, by Eq. \eqref{sep11-1},  $\mbf{L}+\mbf{L}^\ast$ is observable, and thus it is a QND variable. Similar results are also applied to the combined Hermitian coupling operator $\imath(\mbf{L}-\mbf{L}^\ast)(t) \equiv \imath(\mbf{L}-\mbf{L}^\ast)$, in which the QND interaction can be observed by Eq. \eqref{Oct16-1} and the input-output relation \eqref{sep11-2}. Indeed, $\mbf{L}=\mbf{L}^\ast$ assumed in \cite[Sec. 4.2.4]{NY17} is a more special case, which realizes the condition $g=0$ in Eq. \eqref{sepg}. In that special case, $[\mbf{L}+\mbf{L}^\ast,\mbf{H}]=0$ reduces to $[\mbf{L},\mbf{H}]=0$. 
\end{remark}

\bmrk
In the SISO case, let $\mbf{L}$ be an observable and $[\mbf{L},\mbf{H}]=0$, then $\mbf{L}$ is a QND variable and BAE measurement is also realized. If $\mbf{L}$ is not an observable, by Eq. \eqref{eq:sept16_q_out} BAE measurements can still be implemented. However conditions like  $\sum_j |C_{-,j}|^2=\sum_k |C_{+,k}|^2$ should be added to have QND variables.
\emrk

In the following subsections, we focus on the multi-input-multi-output (MIMO) quantum systems and discuss the properties of $[\mbf{L},\mbf{H}]=0$ in the annihilation-creation,  quadrature, Kalman canonical forms, respectively.

\subsection{The MIMO case in the annihilation-creation form}

In the annihilation-creation operator form, $\mbf{L}=\mbf{L}^\#$ is equivalent to  $C_-=C_+^\#$. In particular, if both $C_-$ and $C_+$  are real, then $\mbf{L}=\mbf{L}^\#$ is equivalent to $C_-=C_+$. 

\bmrk
It is obvious that $[\mbf{L},\mbf
{H}]=0$ implies that $[\mbf{L}^\#,\mbf{H}]=0$, and thus $[\mbf{L}\pm\mbf{L}^\#,\mbf{H}]=0$. However, either $[\mbf{L}+\mbf{L}^\#,\mbf{H}]=0$ or  $[\mbf{L}-\mbf{L}^\#,\mbf{H}]=0$ alone does not necessarily lead to $[\mbf{L},\mbf
{H}]=0$.
\emrk

The following lemma will be helpful to derive the equations given in Lemma \ref{Jan29-1}.
\begin{lemma}\cite{ZD22}\label{Jan29-2}
Let $\mbf{X}$, $\mbf{Y}$, and $\mbf{Z}$ be vectors of operators with dimension $l$, $m$, and $n$, respectively. Let $M\in\mathbf{C}^{m\times n}$, and the commutators $[\mbf{a},\mbf{b}]\in\mathbf{C}$ where $\mbf{a}$ and $\mbf{b}$ are arbitrary elements of the vectors $\mbf{X}$, $\mbf{Y}$ and $\mbf{Z}$. Then
\begin{equation}
\left[\mbf{X},\mbf{Y}^\top M \mbf{Z}\right]=\left[\mbf{X},\mbf{Y}^\top\right]M\mbf{Z}+\left[\mbf{X},\mbf{Z}^\top\right]M^\top\mbf{Y}.    
\end{equation} 
\end{lemma}

\begin{lemma}\label{Jan29-1}
$[\mbf{L},\mbf{H}]=0$ if and only if
\begin{equation}\label{Jan29-3}\begin{aligned}
C_-\Omega_-=C_+\Omega_+^\dagger, \\
C_-\Omega_+=C_+\Omega_-^\top,
\end{aligned}\end{equation}
or equivalently, 
\begin{equation}\label{COmega}
\mathcal{C}\Omega=2\Delta(C_-,0)\Omega.
\end{equation}
\end{lemma}

\emph{Proof.}
Inserting $\mbf{L}=C_-\mbf{a}+C_+\mbf{a}^\#$ and $\mbf{H}=\frac{1}{2}\breve{\mbf{a}}^\dagger\Omega\breve{\mbf{a}}$ into the commutator $\left[\mbf{L},\mbf{H}\right]$, by Lemma \ref{Jan29-2} we have
\begin{equation}\la{eq:jun29_LH}
\begin{aligned}
[\mbf{L},\mbf{H}]&=(C_-\Omega_--C_+\Omega_+^\dagger)\mbf{a}+(C_-\Omega_+-C_+\Omega_-^\top)\mbf{a}^\#, \\
[\mbf{L}^\#,\mbf{H}]&=(C_+^\#\Omega_--C_-^\#\Omega_+^\dagger)\mbf{a}+(C_+^\#\Omega_+-C_-^\#\Omega_-^\top)\mbf{a}^\#.
\end{aligned}\end{equation}
Thus, $\left[\mbf{L},\mbf{H}\right]=0$ if and only if \eqref{Jan29-3} holds, and is equivalent to \eqref{COmega} due to the Hermitian property of $\Omega$.




In what follows, we discuss four special cases, which can be used to realize quantum BAE measurements.

\begin{enumerate}

\item $C_+=\mathbf{0}$. In this case,  the conditions in Eq. \eqref{Jan29-1} reduce to $C_-\Omega_-=\mathbf{0}$ and $C_-\Omega_+=\mathbf{0}$, which implies that $\mathcal{C}\Omega=\mathbf{0}$.

As in \cite{GZ15}, for the general linear quantum system, $\Sigma[s]$ is defined as
\begin{equation}
\Sigma[s]=\frac{1}{2}\mathcal{C}(sI+\imath J_n\Omega)^{-1}\mathcal{C}^\flat .
\end{equation}
 Moreover, $\Sigma[s]$ can be further calculated as
\begin{equation}
\Sigma[s]=\frac{1}{2s}\mathcal{C}\sum_{n=0}^\infty(\frac{-\imath}{s})^n(J_n\Omega)^n\mathcal{C}^\flat,
\end{equation}
where $(I+A)^{-1}=\sum_{n=0}^\infty(-1)^n A^n$ has been used in the derivation. We have
\begin{equation}\begin{aligned}
\mathcal{C}J_n\Omega=\left[\begin{array}{cc}C_-\Omega_- & C_-\Omega_+ \\
C_-^\#\Omega_+^\# & -C_-^\#\Omega_-^\#
\end{array}\right]=\mathbf{0},
\end{aligned}\end{equation}
and $\Sigma[s]=\frac{1}{2s}\mathcal{C}\mathcal{C}^\flat$. Consequently, the transfer function reduces to 
\begin{equation}\begin{aligned}
\mathbb{G}[s]=\left[\begin{array}{cc} 
(sI-\frac{1}{2}C_-C_-^\dagger)(sI+\frac{1}{2}C_-C_-^\dagger)^{-1} & 0 \\
0 & (sI-\frac{1}{2}C_-^\#C_-^\top)(sI+\frac{1}{2}C_-^\#C_-^\top)^{-1} 
\end{array}\right],
\end{aligned}\end{equation}
which is block diagonal and BAE measurements are realized. 

\item $C_-=\mathbf{0}$, which means that $\mathcal{C}\Omega=\mathbf{0}$ by Eq. \eqref{COmega}. In this case, we have $C_+\Omega_+^\dagger=0$ and $C_+\Omega_-^\top=0$. Similarly, it can be also verified that $\mathcal{C}J_n\Omega=\mathbf{0}$. And the block diagonal transfer function  
\begin{equation}\begin{aligned}
\mathbb{G}[s]=\left[\begin{array}{cc} 
(sI+\frac{1}{2}C_+C_+^\dagger)(sI-\frac{1}{2}C_+C_+^\dagger)^{-1} & 0 \\
0 & (sI+\frac{1}{2}C_+^\#C_+^\top)(sI-\frac{1}{2}C_+^\#C_+^\top)^{-1} 
\end{array}\right],
\end{aligned}\end{equation}
which realizes BAE measurements.

\item $\Omega_+=\mathbf{0}$. In this case, $C_-\Omega_-=\mathbf{0}$ and $C_+\Omega_-^\top=\mathbf{0}$, which implies that $C_+\Omega_-^\#=0$ and $\mathcal{C}\Omega=\mathbf{0}$. It can be calculated that
$\mathcal{C}J_n\Omega=\mathbf{0}$ and $\Sigma[s]=\frac{1}{2s}\mathcal{C}\mathcal{C}^\flat$. However, the transfer function $\mathbb{G}[s]=(sI-\frac{1}{2}\mathcal{C}\mathcal{C}^\flat)(sI+\frac{1}{2}\mathcal{C}\mathcal{C}^\flat)^{-1}$, which is not block diagonal. Let $C_-C_+^\top=C_+C_-^\top$, i.e., $C_-C_+^\top$ is symmetric, then the transfer function 
\begin{equation}\begin{aligned}
\mathbb{G}[s]=\mathrm{diag}\bigg\{&(sI-\frac{1}{2}(C_-C_-^\dagger-C_+C_+^\dagger))(sI+\frac{1}{2}(C_-C_-^\dagger-C_+C_+^\dagger))^{-1}, \\
&(sI+\frac{1}{2}(C_+^\#C_+^\top-C_-^\#C_-^\top))(sI-\frac{1}{2}(C_+^\#C_+^\top-C_-^\#C_-^\top))^{-1}\bigg\},
\end{aligned}\end{equation}
is block diagonal and BAE measurements are realized.

\item $\Omega_-=\mathbf{0}$, which means $\mathcal{C}\Omega=\mathbf{0}$ by Eq. \eqref{COmega}. Similar to the case of $\Omega_+=\mathbf{0}$, the condition of $C_-C_+^\top$ is symmetric implies that the transfer function $\mathbb{G}[s]$ is block diagonal and BAE measurements are feasible.

\end{enumerate}

\bmrk
The condition $C_-C_+^\top$ being symmetric holds if $C_-=\pm C_+$ ($\mbf{q}$ or $\mbf{p}$ coupling)  or the quantum system is with the SISO case (see Subsection \ref{siso} for details).
\emrk

\bmrk
In the above four special cases with $[\mbf{L},\mbf{H}]=0$, we all have $\mathcal{C}\Omega=\mathbf{0}$, which means that the transfer function $\mathbb{G}[s]$ does not depend on $\Omega$. This may not be true in general even when $[\mbf{L},\mbf{H}]=0$.
\emrk

\subsection{The MIMO case in the quadrature form}

In this section, we investigate linear quantum systems in the quadrature form, whose properties can be proved to be consistent with the related results in the annihilation-creation form given before.

Denote 
\begin{equation}
\Lambda=\left[\begin{array}{cc}C_- & C_+\end{array}\right]V_n^\dagger
=\left[\begin{array}{cc}\Lambda_q & \Lambda_p\end{array}\right].
\end{equation}

In the quadrature representation, the coupling operator $\mbf{L}$ and Hamiltonian $\mbf{H}$ can be described by
\begin{equation}
\mbf{L}=\Lambda \mbf{x}, ~~~~ \mbf{H}=\frac{1}{2}\mbf{x}^\top \mathbb{H} \mbf{x}.
\end{equation}

The condition $[\mbf{L},\mbf{H}]=0$ is equivalent to 
\begin{equation}\label{LH=0real}
\Lambda \mathbb{J}_n \mathbb{H} = 0,
\end{equation}
where Lemma \ref{Jan29-2} has been used in the derivation.

Similarly, $[\mbf{L}^\#,\mbf{H}]=0$ is equivalent to 
\begin{equation}\label{LH=0real2}
\Lambda^\# \mathbb{J}_n \mathbb{H} = 0,
\end{equation}

In  what follows we look at four special cases.

\begin{enumerate}

\item $C_+=0$, then we have $\Lambda_q=\frac{1}{\sqrt{2}}C_-$ and $\Lambda_p=\frac{1}{\sqrt{2}}\imath C_-$. In this case
\begin{equation}
\mathbb{C}=\frac{1}{2}\left[\begin{array}{cc}C_-+C_-^\# & -C_-+C_-^\# \\
C_--C_-^\# & C_-+C_-^\#\end{array}\right], ~~ 
\mathbb{C}^\sharp=-\frac{1}{2}\left[\begin{array}{cc}C_-^\dagger+C_-^\top & -C_-^\dagger+C_-^\top \\
C_-^\dagger-C_-^\top & C_-^\dagger+C_-^\top\end{array}\right].
\end{equation}
Since 
\begin{equation} 
\mathbb{J}_n\mathbb{H}=\left[\begin{array}{cc}
\mathrm{Im}(\Omega_-+\Omega_+) & \mathrm{Re}(\Omega_--\Omega_+) \\
-\mathrm{Re}(\Omega_-+\Omega_+) & \mathrm{Im}(\Omega_--\Omega_+) \\
\end{array}\right],
\end{equation}
we have
\begin{equation}
\Lambda\mathbb{J}_n\mathbb{H}=\left[\begin{array}{cc}\frac{1}{\sqrt{2}}C_- & \frac{1}{\sqrt{2}}\imath C_-\end{array}\right]\mathbb{J}_n\mathbb{H}=\mbf{0},
\end{equation}
which yields
\begin{equation}\label{conclu1}
C_-\Omega_-=0, ~~~~ C_-\Omega_+=0.
\end{equation}
It can be observed that \eqref{conclu1} is as the same as the results given before, and in this case $\mathbb{G}[s]$ is block diagonal, which means that BAE measurements are realized.

\item $C_-=0$, then we have $\Lambda_q=\frac{1}{\sqrt{2}}C_+$ and $\Lambda_p=-\frac{1}{\sqrt{2}}\imath C_+$. In this case
\begin{equation}
\mathbb{C}=\frac{1}{2}\left[\begin{array}{cc}C_++C_+^\# & C_+-C_+^\# \\
C_+-C_+^\# & -C_+-C_+^\#\end{array}\right], ~~ 
\mathbb{C}^\sharp=-\frac{1}{2}\left[\begin{array}{cc}-C_+^\dagger-C_+^\top & C_+^\dagger-C_+^\top \\
C_+^\dagger-C_+^\top & C_+^\dagger+C_+^\top\end{array}\right].
\end{equation}
Since 
\begin{equation} 
\mathbb{J}_n\mathbb{H}=\left[\begin{array}{cc}
\mathrm{Im}(\Omega_-+\Omega_+) & \mathrm{Re}(\Omega_--\Omega_+) \\
-\mathrm{Re}(\Omega_-+\Omega_+) & \mathrm{Im}(\Omega_--\Omega_+) \\
\end{array}\right],
\end{equation}
we have
\begin{equation}
\Lambda\mathbb{J}_n\mathbb{H}=\left[\begin{array}{cc}\frac{1}{\sqrt{2}}C_+ & -\frac{1}{\sqrt{2}}\imath C_+\end{array}\right]\mathbb{J}_n\mathbb{H}=\mbf{0},
\end{equation}
which yields
\begin{equation}\label{conclu2}
C_+\Omega_-^\top=0, ~~~~ C_+\Omega_+^\dagger=0.
\end{equation}
It can be observed that \eqref{conclu2} is as the same as the results given before, and in this case $\mathbb{G}[s]$ is block diagonal, which means that BAE measurements are realized.

\item $\Omega_+=0$, then we have
\begin{equation}
\mathbb{J}_n\mathbb{H}=\left[\begin{array}{cc}\mathrm{Im}(\Omega_-) & \mathrm{Re}(\Omega_-) \\
-\mathrm{Re}(\Omega_-) & \mathrm{Im}(\Omega_-)\end{array}\right].
\end{equation}
Since $\Lambda\mathbb{J}_n\mathbb{H}=0$, we have
\begin{equation}\label{conclu3}
C_-\Omega_-=0, ~~~~ C_+\Omega_-^\top=0.
\end{equation}
It can be observed that \eqref{conclu3} is as the same as the results given before. However, in this case $C_-C_+^\top=C_+C_-^\top$ is needed to satisfy the condition of $\mathbb{G}[s]$ is block diagonal, which means that BAE measurements are realized.

\item $\Omega_-=0$, then we have
\begin{equation}
\mathbb{J}_n\mathbb{H}=\left[\begin{array}{cc}\mathrm{Im}(\Omega_+) & -\mathrm{Re}(\Omega_+) \\
-\mathrm{Re}(\Omega_+) & -\mathrm{Im}(\Omega_+)\end{array}\right].
\end{equation}
Since $\Lambda\mathbb{J}_n\mathbb{H}=0$, we have
\begin{equation}\label{conclu4}
C_-\Omega_+=0, ~~~~ C_+\Omega_+^\#=0.
\end{equation}

It can be observed that \eqref{conclu4} is as the same as the results given before. However, in this case $C_-C_+^\top=C_+C_-^\top$ is also needed to satisfy the condition of $\mathbb{G}[s]$ is block diagonal, which means that BAE measurements are realized.

\end{enumerate}

\subsection{The MIMO case in the Kalman canonical form}

In the case of Kalman canonical form, the condition $[\mbf{L},\mbf{H}]=0$ is equivalent to
\begin{equation}
\Gamma\bar{\mathbb{J}}_n\tilde{\mathbb{H}}=0.
\end{equation} 
Thus, by \cite[Lemma 3.2]{ZPL20} we have
\begin{equation}
\Gamma_{co}\bar{\mathbb{J}}_n\tilde{\mathbb{H}}=
\left[\begin{array}{c}
\Gamma \\
\Gamma^\# 
\end{array}\right]\bar{\mathbb{J}}_n\tilde{\mathbb{H}}=0,    
\end{equation}
which yields 
\begin{equation}
C_hA_h^{22}+C_{co}\mathbb{J}_{n_1}A_{12}^\top
+\frac{1}{2}C_{co}B_{co}\mathbb{J}_mB_h^\top=0,
\end{equation}
and
\begin{equation}\label{Jan13-1}
C_{co}A_{co}=\frac{1}{2}C_{co}B_{co}C_{co}.
\end{equation}
The following theorem presents the realization of BAE measurements by
the linear quantum system in the Kalman canonical form.

\begin{equation}
\Gamma_{co}=\left[\begin{array}{cc}
\Gamma_{co,q} & \Gamma_{co,p} \\
\Gamma_{co,q}^\# & \Gamma_{co,p}^\#
\end{array}\right]
\end{equation}
\begin{equation}
C_{co}=V_m\Gamma_{co}=\left[\begin{array}{c}
C_{co,q} \\
C_{co,p}
\end{array}\right]=\sqrt{2}\left[\begin{array}{cc}
{\rm Re}(\Gamma_{co,q}) & {\rm Re}(\Gamma_{co,p}) \\
{\rm Im}(\Gamma_{co,q}) & {\rm Im}(\Gamma_{co,p})
\end{array}\right]    
\end{equation}

\begin{theorem}
The linear quantum system in the Kalman canonical form under the condition of $[\mbf{L},\mbf{H}]=0$ realizes the BAE measurements of $\mbf{q}_{\rm out}$ with respect to $\mbf{p}_{\rm in}$, i.e., the transfer function $\Xi_{\mbf{p}_{\rm in}\rightarrow \mbf{q}_{\rm out}}(s) \equiv 0$ 
if and only if
\begin{equation}
C_{co,q}B_{co,p}=0,    
\end{equation}
or equivalently,
\begin{equation}
{\rm Re}(\Gamma_{co,q}) {\rm Re}(\Gamma_{co,p}^\top) =
{\rm Re}(\Gamma_{co,p})  {\rm Re}(\Gamma_{co,q}^\top).  
\end{equation}
Meanwhile, the linear quantum system in the Kalman canonical form under the condition of $[\mbf{L},\mbf{H}]=0$ realizes the BAE measurements of $\mbf{p}_{\rm out}$ with respect to $\mbf{q}_{\rm in}$, i.e., the transfer function $\Xi_{\mbf{q}_{\rm in}\rightarrow \mbf{p}_{\rm out}}(s) \equiv 0$ if and only if 
\begin{equation}
C_{co,p}B_{co,q}=0,    
\end{equation}
or equivalently,
\begin{equation}
{\rm Im}(\Gamma_{co,q}) {\rm Im}(\Gamma_{co,p}^\top) =
{\rm Im}(\Gamma_{co,p})  {\rm Im}(\Gamma_{co,q}^\top).  
\end{equation}
Furthermore, the realization of both the two cases of BAE measurements imply that ${\rm Re}(\Gamma_{co,q}\Gamma_{co,p}^\top)$ is symmetric.
\end{theorem}

\begin{proof}
Partition the system matrices $C_{co}$ and $B_{co}$ as
\begin{equation}
C_{co}=\left[\begin{array}{c}
C_{co,q} \\
C_{co,p}
\end{array}\right], ~~~~ B_{co}=\left[\begin{array}{cc}
B_{co,q} & B_{co,p} 
\end{array}\right],   
\end{equation}
respectively. Under the condition $[\mbf{L},\mbf{H}]=0$, by \eqref{Jan13-1} we have
\begin{equation}\begin{aligned}
C_{co}A_{co}^kB_{co}=\frac{1}{2^k}(C_{co}B_{co})^{k+1}.
\end{aligned}\end{equation}
Thus, if $C_{co,q}B_{co,p}=0$, then $C_{co}A_{co}^kB_{co}$ is block lower triangular, which implies that $C_{co,q}A_{co}^kB_{co,p}=0$. By \cite[Theorem 4.1]{ZPL20}, the quantum Kalman canonical form realizes the BAE measurements of $\mbf{q}_{\rm out}$ with respect to $\mbf{p}_{\rm in}$. On the other hand, if the BAE measurements of $\mbf{q}_{\rm out}$ with respect to $\mbf{p}_{\rm in}$ is realized, then by \cite[Theorem 4.1]{ZPL20}
\begin{equation}
C_{co,q}A_{co}^kB_{co,p}=0, ~~ k=0,1,\ldots,
\end{equation}
which implies that $C_{co,q}B_{co,p}=0$. The proof of the BAE measurements of $\mbf{p}_{\rm out}$ with respect to $\mbf{q}_{\rm in}$ follows in a similar way. 
\end{proof}

\begin{example}
In \cite[Fig. 1(A)]{LOW+21}, it can be calculated that

\begin{equation}
C_{co}=\left[\begin{array}{c}
C_{co,q} \\
C_{co,p} 
\end{array}\right]=\left[\begin{array}{cc}
\sqrt{\kappa} & 0 \\
0 & \sqrt{\kappa}
\end{array}\right],  ~~  
B_{co}=\left[\begin{array}{cc}
B_{co,q} & B_{co,p} 
\end{array}\right]=-\left[\begin{array}{cc}
\sqrt{\kappa} & 0 \\
0 & \sqrt{\kappa}
\end{array}\right],
\end{equation}
which imply that
\begin{equation}
\left\{\begin{array}{c}
C_{co,q}B_{co,p}=0, \\
C_{co,p}B_{co,q}=0.
\end{array}\right.     
\end{equation}
Thus, both the BAE measurements of $\mbf{q}_{\rm out}$ with respect to $\mbf{p}_{\rm in}$ and $\mbf{p}_{\rm out}$ with respect to $\mbf{q}_{\rm in}$ are realized in this experiment. 
\end{example}

\bmrk
According to  \cite[Lemma 3.2]{ZPL20}, if $H$ is of the form of  \cite[Eq. (12)]{ZPL20} and $\Gamma $ is of the form of  \cite[Eq. (13)]{ZPL20}, then the linear system is of the Kalman canonical form. If $\Gamma_h \neq 0$, then the variables $p_h$ exist, which means there are QND variables.
\emrk

\subsection{Realization of QND variables under QND interaction}


QND variables correspond to the uncontrollable and observable subsystems. In this subsection, we investigate sufficient conditions for the realization of a QND observable under QND interaction, i.e., $[\mbf{L},\mbf{H}]=0$.

Since $\Lambda=\frac{1}{\sqrt{2}}\left[\begin{array}{cc}C_-+C_+ & \imath(C_--C_+)\end{array}\right]=\left[\begin{array}{cc}\Lambda_q & \Lambda_p\end{array}\right]$, and
\begin{equation}
\mathbb{B}=-\left[\begin{array}{cc}\mathrm{Re}(C_-^\dagger-C_+^\dagger) & -\mathrm{Im}(C_-^\dagger-C_+^\dagger) \\
\mathrm{Im}(C_-^\dagger+C_+^\dagger) & \mathrm{Re}(C_-^\dagger+C_+^\dagger)\end{array}\right]=
-\sqrt{2}\left[\begin{array}{cc}-\mathrm{Im}(\Lambda_p^\dagger) & -\mathrm{Re}(\Lambda_p^\dagger) \\
\mathrm{Im}(\Lambda_q^\dagger) & \mathrm{Re}(\Lambda_q^\dagger)\end{array}\right],
\end{equation}
we need $\Lambda_q=0$ or $\Lambda_p=0$ to realize QND variables.

\begin{itemize}

\item $\Lambda_q=0$, i.e., $C_-=-C_+$. In this case, $C_--C_+=2C_-=-\sqrt{2}\imath \Lambda_p$, and $\Lambda_p=\sqrt{2}\imath C_-$. We have
\begin{equation}
\mathbb{C}=2\left[\begin{array}{cc}0 & -\mathrm{Im}(C_-) \\
0 & \mathrm{Re}(C_-)\end{array}\right], ~~~~ \mathbb{B}=
2\left[\begin{array}{cc}-\mathrm{Re}(C_-^\dagger) & \mathrm{Im}(C_-^\dagger) \\
0 & 0\end{array}\right],
\end{equation}
and 
\begin{equation}
\mathbb{C}^\sharp\mathbb{C}=4\left[\begin{array}{cc}0 & -\mathrm{Re}(C_-^\dagger)\mathrm{Im}(C_-)-\mathrm{Im}(C_-^\dagger)\mathrm{Re}(C_-) \\
0 & 0\end{array}\right].
\end{equation}
If further $\Omega_-=-\Omega_+$, then $\Omega_--\Omega_+=2\Omega_-$, and 
\begin{equation}
\mathbb{A}=2\left[\begin{array}{cc}0 & \mathrm{Re}(\Omega_-) \\
0 & \mathrm{Im}(\Omega_-)\end{array}\right]-2\left[\begin{array}{cc}0 & -\mathrm{Re}(C_-^\dagger)\mathrm{Im}(C_-)-\mathrm{Im}(C_-^\dagger)\mathrm{Re}(C_-) \\
0 & 0\end{array}\right].
\end{equation}
Actually, $C_-=-C_+$ ($\mbf{p}$ coupling) and $\Omega_-=-\Omega_+$ imply that $[\mbf{L},\mbf{H}]=0$. In this case, we have
\begin{equation}\begin{aligned}
\dot{\mbf{q}}&=2\left[\mathrm{Re}(\Omega_-)+\mathrm{Re}(C_-^\dagger)\mathrm{Im}(C_-)+\mathrm{Im}(C_-^\dagger)\mathrm{Re}(C_-)\right]\mbf{p}-2\mathrm{Re}(C_-^\dagger)\mbf{q}_{\mathrm{in}}+2\mathrm{Im}(C_-^\dagger)\mbf{p}_{\mathrm{in}}, \\
\dot{\mbf{p}}&=2\mathrm{Im}(\Omega_-)\mbf{p}, \\
\mbf{q}_{\mathrm{out}}&=-2\mathrm{Im}(C_-)\mbf{p}+\mbf{q}_{\mathrm{in}}, \\
\mbf{p}_{\mathrm{out}}&=2\mathrm{Re}(C_-)\mbf{p}+\mbf{p}_{\mathrm{in}}.
\end{aligned}\end{equation}

Thus, $\mbf{p}$ is a QND variable if the subsystem $(\mathrm{Im}(\Omega_-),0,-\mathrm{Im}(C_-))$ or $(\mathrm{Im}(\Omega_-),0,\mathrm{Re}(C_-))$ is observable. Moreover, BAE measurement is realized from $\mbf{p}_{\mathrm{in}}$ to $\mbf{q}_{\mathrm{out}}$ as well as from $\mbf{q}_{\mathrm{in}}$ to $\mbf{p}_{\mathrm{out}}$.

\item $\Lambda_p=0$, i.e., $C_-=C_+$. In this case, $C_-+C_+=2C_-=\sqrt{2}\Lambda_q$, and $\Lambda_q=\sqrt{2}C_-$. We have
\begin{equation}
\mathbb{C}=2\left[\begin{array}{cc}\mathrm{Re}(C_-) & 0 \\
\mathrm{Im}(C_-) & 0\end{array}\right], ~~~~ \mathbb{B}=
-2\left[\begin{array}{cc}0 & 0 \\
\mathrm{Im}(C_-^\dagger) & \mathrm{Re}(C_-^\dagger)\end{array}\right],
\end{equation}
and 
\begin{equation}
\mathbb{C}^\sharp\mathbb{C}=4\left[\begin{array}{cc}0 & 0 \\
\mathrm{Im}(C_-^\dagger)\mathrm{Re}(C_-)+\mathrm{Re}(C_-^\dagger)\mathrm{Im}(C_-) & 0\end{array}\right].
\end{equation}

If further $\Omega_-=\Omega_+$, then $\Omega_-+\Omega_+=2\Omega_-$, and 
\begin{equation}
\mathbb{A}=2\left[\begin{array}{cc}\mathrm{Im}(\Omega_-) & 0 \\
-\mathrm{Re}(\Omega_-) & 0\end{array}\right]-2\left[\begin{array}{cc}0 & 0 \\
\mathrm{Im}(C_-^\dagger)\mathrm{Re}(C_-)+\mathrm{Re}(C_-^\dagger)\mathrm{Im}(C_-) & 0\end{array}\right].
\end{equation}
Actually, $C_-=C_+$ ($\mbf{q}$ coupling) and $\Omega_-=\Omega_+$ also imply that $[\mbf{L},\mbf{H}]=0$. In this case, we have
\begin{equation}\begin{aligned}
\dot{\mbf{q}}&=2\mathrm{Im}(\Omega_-)\mbf{q}, \\
\dot{\mbf{p}}&=-2\left[\mathrm{Re}(\Omega_-)+\mathrm{Im}(C_-^\dagger)\mathrm{Re}(C_-)+\mathrm{Re}(C_-^\dagger)\mathrm{Im}(C_-)\right]\mbf{q}-2\mathrm{Im}(C_-^\dagger)\mbf{q}_{\mathrm{in}}-2\mathrm{Re}(C_-^\dagger)\mbf{p}_{\mathrm{in}}, \\
\mbf{q}_{\mathrm{out}}&=2\mathrm{Re}(C_-)\mbf{q}+\mbf{q}_{\mathrm{in}}, \\
\mbf{p}_{\mathrm{out}}&=2\mathrm{Im}(C_-)\mbf{q}+\mbf{p}_{\mathrm{in}}.
\end{aligned}\end{equation}
Thus, $\mbf{q}$ is a QND variable if the subsystem $(\mathrm{Im}(\Omega_-),0,\mathrm{Im}(C_-))$ or $(\mathrm{Im}(\Omega_-),0,\mathrm{Re}(C_-))$ is observable. Moreover, BAE measurement is realized from $\mbf{p}_{\mathrm{in}}$ to $\mbf{q}_{\mathrm{out}}$ as well as from $\mbf{q}_{\mathrm{in}}$ to $\mbf{p}_{\mathrm{out}}$.

\end{itemize}


Finally, if $C_+=0$ and $C_-$ is real, then 
\beq
\mathbb{B}=-\left[\begin{array}{cc}
C_-^\dagger & 0 \\
0 & C_-^\dagger
\end{array}\right], \ \ \mathbb{C}=
\left[
\bey{@{}ll@{}}
C_- & \\
 & C_-
\eey
\right],
\eeq
and
\beq
\mathbb{A}=\left[
\begin{array}{cc}
 \mathrm{Im}(\Omega_-+\Omega_+) & \mathrm{Re}(\Omega_--\Omega_+) \\
 -\mathrm{Re}(\Omega_-+\Omega_+) & \mathrm{Im}(\Omega_--\Omega_+) \\
 \end{array}
  \right]-\frac1{2}\left[
  \begin{array}{cc}
C_-^\dagger C_- & 0 \\
0 & C_-^\dagger C_-
\end{array}\right].
\eeq
Furthermore, if $\Omega_-=\Omega_+$, then
\beq
\mathbb{A}=2\left[\begin{array}{cc}\mathrm{Im}(\Omega_-) & 0 \\
-\mathrm{Re}(\Omega_-) & 0\end{array}\right]-\frac1{2}\left[
  \begin{array}{cc}
C_-^\dagger C_- & 0 \\
0 & C_-^\dagger C_-
\end{array}\right].
\eeq
Thus, BAE measurement can be realized. However, there is no QND variables.

\section{Conclusion}\label{Conclu}

This paper establishes a comprehensive framework for realizing BAE measurements and QND variables in linear quantum systems. We demonstrate that bilateral BAE measurements of conjugate observables can be achieved when the system Hamiltonian is purely imaginary and the coupling operator is real or purely imaginary. Furthermore, specific coupling configurations (e.g., $C_-=C_+$ or $C_-=-C_+$) yield QND variables for position or momentum, respectively. The analysis extends to systems where direct BAE realization conditions are not met, showing that coherent feedback control can engineer an effective Hamiltonian to enable BAE measurements. Crucially, under QND interaction condition $[\mbf{L},\mbf{H}]=0$, the coupling operator itself becomes a QND observable, allowing for continuous, non-perturbative monitoring. These results, formulated in both annihilation-creation and quadrature operator representations and clarified through the Kalman canonical form, provide a systematic design methodology for high-precision quantum metrology, with direct applications in quantum sensing, gravitational wave detection, and quantum information processing. Future work may explore nonlinear extensions and non-Markovian environments to further advance quantum measurement and control capabilities.


\medspace
\textbf{Acknowledgment.} This work is partially financially supported by Quantum Science and Technology-National Science and Technology Major Project 2023ZD0300600, Guangdong Provincial Quantum Science Strategic Initiative No. GDZX2303007, Hong Kong Research Grant Council (RGC) under Grant No. 15213924, and  the CAS AMSS-PolyU Joint Laboratory of Applied Mathematics.

\bibliographystyle{plain}
 \bibliography{gzhang}

\end{document}